\documentclass[a4paper,10pt]{article}
\usepackage[utf8]{inputenc}
\usepackage[hidelinks,hyperindex,breaklinks]{hyperref}
\usepackage{graphicx}
\usepackage[english]{babel}
\usepackage{amsmath,amssymb,amsthm,dsfont}
\usepackage{enumitem}
\usepackage{todonotes}
\usepackage[natbibapa]{apacite} 
\usepackage{authblk}
\usepackage{tikz}

\theoremstyle{plain}
\newtheorem{theorem}{Theorem}[section]
\newtheorem{proposition}[theorem]{Proposition}
\newtheorem{corollary}[theorem]{Corollary}
\newtheorem{lemma}[theorem]{Lemma}

\theoremstyle{definition}
\newtheorem{definition}[theorem]{Definition}
\newtheorem{assumption}[theorem]{Assumption}
\newtheorem{example}[theorem]{Example}
 
\theoremstyle{remark}
\newtheorem{remark}[theorem]{Remark}

\usepackage{color,tabto}

\numberwithin{equation}{section}

\newcommand{\R}{\mathbb{R}}
\newcommand{\E}{\mathds{E}}
\newcommand{\F}{\mathcal{F}}
\newcommand{\ind}{\mathds{1}}

\DeclareMathOperator{\Var}{Var}
\DeclareMathOperator{\diag}{diag}

\title{A structural Heath-Jarrow-Morton framework for consistent intraday, spot, and futures electricity prices}

\author[1,2]{W.J. Hinderks\thanks{Corresponding author: \url{wieger.hinderks@itwm.fraunhofer.de}}}
\author[1]{A. Wagner}%\thanks{\url{andreas.wagner@itwm.fraunhofer.de}}}
\author[1,2]{R. Korn}%\thanks{\url{korn@mathematik.uni-kl.de}}}
\affil[1]{\small Fraunhofer ITWM, Fraunhofer-Platz 1, 67663 Kaiserslautern, Germany}
\affil[2]{\small TU Kaiserslautern, Erwin-Schrödinger-Straße 1, 67663 Kaiserslautern, Germany}

\date{\today}

\begin{document}

\maketitle

\begin{abstract}
\noindent In this paper we introduce a flexible HJM-type framework that allows for consistent modelling of intraday, spot, futures, and option prices. This framework is based on stochastic processes with economic interpretations and consistent with the initial term structure given in the form of a price forward curve. Furthermore, the framework allows for existing day-ahead spot price models to be used in an HJM setting. We include several explicit examples of classical spot price models but also show how structural models and factor models can be formulated within the framework.

\vspace{1em}

\end{abstract}

%Questions / main ideas of this paper:
%\begin{enumerate}
%\item Why would use an HJM approach instead of direct spot price modelling? $\to$ motivation
%\item We want to offer a HJM framework in which \emph{everything} has an economic interpretation, that allows for spot, futures and option valuation, and is consistent with the current forward term structure (the price forward curve)
%\end{enumerate}
%\textbf{TODO} Clearly list the contributions!

%%%%%%%
%In this paper we claim three (four) things: 
%\begin{enumerate}
%\item extend to intraday
%\item Q = P
%\item spot is forward contract
%\item practical, economic intuition
%\end{enumerate}
%%%%%%%

%===== INTRODUCTION =====%
\section{Introduction} \label{section:Introduction}

In recent years the electricity intraday markets have gained increased popularity: the traded volume at the German/Austrian intraday market has grown by 30.3 percent from May~2016 to May~2018 \citep{EPEX_PR20170601,EPEX_PR20180601}. Since different electricity contracts exhibit different price behaviour such as spikes in the day-ahead spot but not in futures prices, it is a rising challenge in energy finance to define a single model that allows for a joint simulation of power prices at intraday, spot, and futures markets.

In this paper we suggest a Heath-Jarrow-Morton framework for modelling electricity prices. The framework is consistent with the current forward term structure (i.e. the price forward curve) and we motivate each mathematical component by an economic interpretation. Furthermore, we discuss the computation of intraday, spot, and futures prices within this framework and we show how options on futures contracts can be priced. A new approach is the use of structural models for day-ahead spot price modelling within a Heath-Jarrow-Morton framework.

The starting point for a Heath-Jarrow-Morton~(HJM\footnote{See \cite{Heath1992} for the original paper introducing this framework for interest rate modelling.}) approach for electricity prices is the fictitious \emph{forward price} or \emph{forward kernel}.\footnote{Forward kernel is the name used by \cite{Caldana2017}.} The forward kernel~$f_t(\tau)$, $t \leq \tau$, is the price at time~$t$ of a forward contract delivering electricity instantly at time~$\tau$. It follows that the price at $t$ of a futures contract delivering from $\tau_1$ to $\tau_2$ is the averaged forward kernel during the delivery period, i.e.
\begin{equation} \label{eq:introductionfuturesprice}
F_t(\tau_1, \tau_2) =  \frac{1}{\tau_2 - \tau_1} \, \int_{\tau_1}^{\tau_2} f_t(u) \, du, \quad t \leq \tau_1.
\end{equation}
In the HJM framework for interest rates the forward rate is modelled instead of the short rate (cf. \cite{Brigo2006}). Therefore, modelling the forward kernel instead of the day-ahead spot price\footnote{Modelling of the day-ahead spot price is a common approach, for which several different approaches have been developed, cf. \cite{Weron2014}.} makes this an HJM approach for power prices. Furthermore, just like in the HJM framework for interest rates, the forward kernel itself is not a traded product at the market but its (integrated) derivatives are.

Several models for the forward kernel~$f_t(\tau)$ have been introduced by \cite{Clewlow1999,BenthKoekebakker2008,Kiesel2009,Hinz2005,Koekebakker2005}. They define the forward kernel dynamics driven by Brownian motions. However, since the day-ahead spot prices show spikes, these models have drawbacks. Therefore, there is a need for a forward kernel model that allows for spikes in relatively short delivery periods (day-ahead spot contracts) but smooths these out for longer delivery periods (futures contracts). The theoretical HJM framework of \cite{BenthPiccirilli2017a} introduces forward kernel dynamics driven by Brownian motions and pure jump L\'evy processes. However, \cite{BenthPiccirilli2017a} assume that day-ahead spot and futures contracts are priced under two different measures. Motivated by economic arguments, this ambiguity is avoided in our approach.

In the literature the use of more than one probability measure has also been challenged: \cite{Lyle2009,Caldana2017} assume a single probability measure, for example. This is supported by the fact that it is not clear which equivalent measure should be the pricing measure~$Q$. Since electricity is a non-storable commodity and buy-and-hold strategy arguments are not valid, it is not clear what the relation between the price of electricity contracts and the money market account is \citep{Bessembinder2002}. This also implies that the market is incomplete and that there are (possibly) infinitely many equivalent martingale measures. Again, this leaves the choice of pricing measure unclear.

\begin{figure}[t]
\centering
\begin{tikzpicture}
\draw[gray,thick] (2,0.1) -- (2,-0.1);
\node[above] at (2,0.1) {\tiny$d-1$};

\draw[gray,thick] (5.5,0.1) -- (5.5,-0.1);
\node[above] at (5.5,0.1) {\tiny$d$};

\draw[gray,thick] (9,0.1) -- (9,-0.1);
\node[above] at (9,0.1) {\tiny$d+1$};

\draw[red,thick] (2.8,0.1) -- (2.8,-0.1);
\node[above] at (2.8,0.1) {\rotatebox{90}{\tiny 10:00 EXAA}};

\draw[red,thick] (3.5,0.1) -- (3.5,-0.1);
\node[above] at (3.5,0.1) {\rotatebox{90}{\tiny 12:00 EPEX}};

%\draw[blue] (2,-0.1) -- (4.2,-0.1);
\node[below] at (3.1,-0.1) {\tiny EPEX day-ahead};

\draw[thick,red] (4.2,0.1) -- (4.2,-0.1);
\node[above] at (4.2,0.1) {\rotatebox{90}{\tiny 15:00 EPEX}};
\node[above] at (4.5,0.1) {\rotatebox{90}{\tiny quarter hour}};
\draw[red,dashed] (4.2,-0.1) -- (8.1,-0.1);
\node[below] at (6.15,-0.1) {\tiny EPEX intraday};

\draw[thick] (8.1,0.1) -- (8.1,-0.1);
\node[above] at (8.1,0.1) {$\tau$};

\draw[densely dotted,thick,->] (-0.5,0) -- (0,0);
\draw[thick,->] (0,0) -- (10,0);
\node[above] at (0,0.1) {$t$};

\draw[red,dashed] (-0.5,-0.1) -- (2,-0.1);
\node[below right,text width=2cm] at (0,-0.1) {\tiny EEX futures};

\node[] at (4.75,2) {\textbf{Observation structure of $f_t(\tau)$}};
\end{tikzpicture}
\caption{Observation structure of~$f_t(\tau)$ for the German electricity market and a fixed delivery time~$\tau$. The red marked lines and time points are the (indirect) observation moments. The lines with~$d-1$ and $d$ stand for the start of day~$d-1$ and $d$.}
\label{fig:marketdevelopmentofforwardprice}
\end{figure}
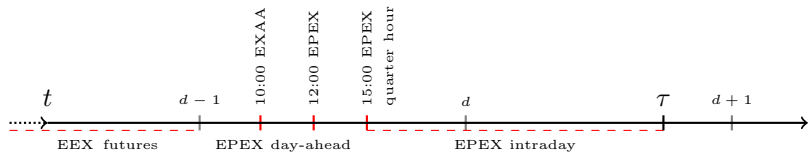

We follow the idea of \cite{Caldana2017} that the prices of day-ahead spot and futures contracts both should be computed by Equation~\eqref{eq:introductionfuturesprice}. This actually sounds intuitively since, for example at the German markets, day-ahead spot contracts are traded at least twelve hours before delivery. In other countries such as the US the terminology is different: the day-ahead spot price is commonly referred to as the forward price~\citep{Longstaff2004}. Even in Europe, with the increasing popularity of the intraday markets, we observe a shift in terminology: \cite{Weron2014} remarks that the term \emph{spot} is used more and more frequently for the real-time or intraday market. We will always explicitly state to which spot market we refer.

In this paper we even propose to extend Equation~\eqref{eq:introductionfuturesprice} to the intraday market. Figure~\ref{fig:marketdevelopmentofforwardprice} gives an example of the development of the forward kernel~$f_t(\tau)$ and how it becomes observable at the German/Austrian market. First the forward kernel~$f_t(\tau)$ is only (partly) observable through EEX~futures contracts. Then the Austrian EXAA and two German EPEX day-ahead spot auctions\footnote{Of course, it is not clear what the roles the EXAA and EPEX will play for each other after the announced market division. Press release: \url{https://www.bundesnetzagentur.de/SharedDocs/Pressemitteilungen/EN/2017/15052017_DE_AU.html} (visited on March 26 2018).} are held, after which the EPEX intraday spot market opens.
%This last market is the closest approximation to $f_t(\tau)$. As a matter of fact, if we discretize the delivery time~$\tau$ in 15 minute intervals, we can approximate $f_t(\tau)$ by the intraday spot price. We will use the intraday spot price as a direct approximation for the forward kernel.

Furthermore, we show how the classical models described by \cite{Schwartz2000,Lucia2002} fit into our framework. We also show how other more general day-ahead spot price models can be used to fit into our model. A particular new example we introduce in this paper, is to use structural models in the context of an HJM framework. We also apply our framework to the setting of multi-factor models.

This paper is structured as follows: Section~\ref{section:HJMframework} introduces a model for the forward kernel based on the economic intuition that there are two driving components behind the forward kernel. The first component is the equilibrium of supply and demand at delivery time and the second is a general noise from partially informed traders or illiquidity at trading time~$t$. Successively, in Section~\ref{section:futures} and Section~\ref{section:options} the futures and option prices are computed, respectively. Section~\ref{section:applications} contains the above explicitly mentioned examples for the market equilibrium process, while Section~\ref{section:Conclusion} concludes.

%===== MODEL =====%
\section{Heath-Jarrow-Morton framework} \label{section:HJMframework}
In Section~\ref{section:instaneousforwardpricetheory} we will define a model for the forward kernel motivated by economic interpretations. Using this model in Sections~\ref{section:futures} and \ref{section:options} we derive the prices of futures contracts and options on futures contracts, respectively. Section~\ref{section:OverviewContracts} gives an overview of the prices for different electricity contracts for the example of the German market.

%===== FORWARD PRICE =====%
\subsection{Forward kernel} \label{section:instaneousforwardpricetheory}
The forward kernel~$f_t(\tau)$ is the price at time~$t$ of a forward contract delivering 1~MW instantly at time~$\tau$. Throughout the rest of this paper we interpret~$t$ as the \emph{trading time} and $\tau$ as the \emph{delivery time}.

For $\tau \geq 0$ let $X^\tau = \{ X_t^\tau ; t \geq 0\}$ and $Y = \{ Y_t; t \geq 0 \}$ be two independent, a.s. c\`adl\`ag stochastic processes on the complete probability space~$(\Omega, \F, P)$ taking values in $\R$ and $\R^n$, respectively. Furthermore, assume that the processes $X^\tau$ for each $\tau \geq 0$ and $Y$ are adapted to the filtration~$\{\F_t ; t\geq 0 \}$, which satisfies the usual conditions, i.e. $\{\F_t ; t\geq 0 \}$ is right-continuous and $\F_0$ contains all $P$-null sets. The filtration generated by $Y$ and $X^\tau$ augmented by all $P$-null sets automatically fulfills these conditions. Finally, let $g \, : \, \R^n \to \R$ be a function such that $g(Y_t)$ is real-valued stochastic process.

We have two strong economic interpretations for these two stochastic processes: we interpret the $n$-dimensional process~$Y_t$ as the randomness or the state of the market, where each component of $Y_t$ stands for a (random) facet of the market, e.g. demand, load, or weather predictions. The function~$g$ maps the state of the market state~$Y_t$ to its corresponding price. Combining the fact that our inspiration came from the class of structural models for day-ahead spot price modelling and the fact that it gives the basic structure to the forward kernel, we call the pair~$(g, Y_t)$ the \emph{structural component}. Often we will also only call~$Y_t$ the structural component. 

The process~$X_t^\tau$ is called the \emph{market noise} because it accounts for the incomplete market information of all market participants and illiquidity of the market. An example of incomplete market information is the uncertainty of weather predictions: nobody knows with complete certainty about the future weather or temperature. With these interpretations we define the forward kernel:

\begin{definition}[Forward kernel] \label{definition:instantaneousforward}
We define the \emph{forward kernel} at trading time~$t$ and delivery time~$\tau$ as 
\[
f_t(\tau) :=  X_t^\tau \, \E\left[ g(Y_\tau) \, | \, \F_t \right],
\]
where $X_t^\tau$ is the \emph{market noise} at trading time~$t$ for the delivery time~$\tau$ and $(g, Y_\tau)$ the \emph{structural component} at delivery time~$\tau$.
\end{definition}

%\begin{remark}[Consistency with spot price models] \label{remark:consistencywithspotpricemodels}
%It is clear that the forward kernel is consistent with day-ahead spot price~$S_t$ models by setting $X_t^\tau := 1$ and choosing the identity for~$g$ with $Y_t := S_t$.We will come back to this topic in Section~\ref{section:applications}, where we will show how the models of \cite{Schwartz2000} and \cite{Lucia2002} fit in our framework.
%\end{remark}

We use the notation~$X_t^\tau$ to emphasize that the market noise is a stochastic process in the trading time~$t$ but can (deterministically) depend on the delivery time~$\tau$, whereas the structural component~$Y_\tau$ only depends on delivery time. Economically, this makes sense since the imbalance of supply and demand at delivery time~$\tau$ determines the price independent of the trading time~$t$ at which we predict this imbalance. However, the market noise is the disturbance of this prediction originating from market participants with incomplete market information, which intuitively depends on both the trading time~$t$ and the delivery time~$\tau$ they are trying to predict.\footnote{This also allows for seasonal volatility in the market noise.} Although we call $X_t^\tau$ the market noise, it can also be interpreted as a measure transformation (or Radon-Nikodym derivative, see Remark~\ref{remark:ChangeOfMeasure}) or as a general additional component that introduces an additional degree of freedom in the modelling process.

\begin{assumption}[Market noise] \label{assumption:marketnoise}
The process~$X^\tau = \{ X_t^\tau ; t \geq 0\}$ with its interpretation as market noise for delivery time~$\tau$ is defined as multiplicative stochastic noise. We assume that it is an a.s. positive c\`adl\`ag martingale with expectation one, i.e. $\E X_t^\tau = 1$ for all $\tau \geq t \geq 0$. In particular, we assume that the initial value $X_0^\tau = 1$ a.s. for all $\tau \geq 0$. 
\end{assumption}

\begin{assumption}[Structural component] \label{assumption:marketimbalance}
We assume that $Y = \{ Y_t; t\geq 0\}$ is a $\R^n$-valued c\`adl\`ag stochastic process. In particular, we assume that the initial value equals $Y_0 = y_0 \in \R^n$ a.s. such that $g(y_0) = f_0(0)$, where~$f_0(\tau)$ is the price forward curve~(PFC) for delivery time~$\tau$, which we assume to be known (cf. Remark~\ref{remark:PFCConstruction}). Furthermore, as a technical assumption we need that $\E |g(Y_t)| < \infty$ for all $t\geq 0$. Finally, although we assume that $g(Y_t)$ can take all values in~$\R$, including negative values, we assume that its expectation~$\E g(Y_t) > 0$ is strictly positive. The economic interpretation behind this assumption is that we do not \emph{expect} negative forward prices to occur.
\end{assumption}

With these assumptions the sign of the forward kernel is uniquely determined by the structural component~$Y$ and the process $X^\tau$ cannot influence it. Furthermore, the expectation $\E f_t(\tau)$ is fully determined by the structural component~$Y_\tau$ and independent of trading time~$t$ (cf. Lemma~\ref{lemma:expectationinstantaneousforward}).

\begin{remark}[Price forward curve~$f_0(\tau)$] \label{remark:PFCConstruction}
In the framework the price forward curve~(PFC), denoted by~$f_0(\tau)$, plays an important role: it determines the expectation of the forward kernel~$f_t(\tau)$. There are many studies that describe how one can construct a PFC from market prices such as \cite{Caldana2017,Kiesel2018}, for example. In practice every energy utility has an in-house PFC. In the following we will therefore assume that the PFC is known.
\end{remark}

\begin{theorem}
For fixed~$\tau \geq 0$ the forward kernel process~$f(\tau) := \{ f_t(\tau); t \geq 0 \}$ is an adapted stochastic process. Furthermore, $f(\tau)$ is a.s. c\`adl\`ag.
\end{theorem}
\begin{proof}
By definition~$f(\tau)$ is a stochastic process. Moreover, since we assumed $X_t^\tau$ to be $\F_t$-measurable and since the conditional expectation~$Z_t^\tau := \E[g(Y_\tau) \, | \, \F_t ]$ is always $\F_t$-measurable, the $\F_t$-measurability of $f_t(\tau)$ follows immediately. Because the filtration satisfies the usual conditions, $Z_t^\tau$ has a c\`adl\`ag modification~\cite[Chapter 1, Theorem 3.13]{Karatzas1998}. Since the conditional expectation~$Z_t^\tau$ is uniquely defined up to null sets, we can choose this modification and the result follows by the assumption that $X_t^\tau$ is c\`adl\`ag.
\end{proof}

Since we assume that $X^\tau$ and $Y$ both a.s. start at a deterministic value, we assume without loss of generality that $\F_0$ is generated by $\Omega$ and all $P$-null sets. This in particular implies that $\E g(Y_\tau) =  \E\left[ g(Y_\tau) \, | \, \F_0 \right]$, a fact we will exploit in the next lemma.

\begin{lemma} \label{lemma:expectationinstantaneousforward}
For fixed~$\tau \geq 0$ the forward kernel process~$f(\tau) := \{ f_t(\tau); t \geq 0 \}$ is a martingale. Furthermore, its expectation is given by
\[
\E f_t(\tau) = \E g(Y_\tau) =  f_0(\tau)
\]
for all $0 \leq t \leq \tau$.
\end{lemma}
\begin{proof}
The product of two independent martingales clearly is a martingale. Furthermore, it follows immediately from Assumption~\ref{assumption:marketnoise} and \ref{assumption:marketimbalance} that
\[
\E f_t(\tau) = \E\left[X_t^\tau\right] \, \E\left[\E\left[ g(Y_\tau) \, | \, \F_t \right]\right] = \E g(Y_\tau) =  X_0^\tau \, \E\left[ g(Y_\tau) \, | \, \F_0 \right] =  f_0(\tau) 
\]
by the independence of $X_t^\tau$ and $Y_t$.
\end{proof}

Lemma~\ref{lemma:expectationinstantaneousforward} also imposes a condition for the expectation~$\E g(Y_\tau)$ of the structural component, which can be used to calibrate the structural component~$Y$ and function~$g$ after the PFC~$f_0(\tau)$ has been determined. If one wants to obtain a model that is consistent with an existing PFC~$f_0(\tau)$, one needs to choose and calibrate $g$ and $Y$ such that $\E g(Y_\tau) = f_0(\tau)$. 

\begin{remark}[Change of measure] \label{remark:ChangeOfMeasure}
In the previous discussion we considered the measure space~$(\Omega, \F, P)$ equipped with the real-world measure~$P$. However, in arbitrage-free markets there is a pricing measure under which derivatives are valued. The \emph{$\tau$-forward measure}~$Q^\tau$ defined by its Radon-Nikodym derivative
\begin{equation} \label{eq:PricingMeasureDefinedByX}
\frac{dP}{dQ^\tau}\Big|_{\F_t} = X_t^\tau
\end{equation}
could be used for this purpose. Using the $\tau$-forward measure and Bayes' theorem for conditional expectations we can rewrite Definition~\ref{definition:instantaneousforward}
\[
f_t(\tau) = X_t^\tau \, \E_P[ g(Y_\tau) \, | \, \F_t] = \E_{Q^\tau}[ X_\tau^\tau \, g(Y_\tau) \, | \, \F_t].
\]
The latter term can be defined
\[
S_t := X_t^t \, g(Y_t),
\]
which yields a general spot price model. The choice of the stochastic process~$X_t^\tau$ can be viewed as the choice of a pricing measure~$Q^\tau$ in light of Equation~\eqref{eq:PricingMeasureDefinedByX}. If the noise~$X_t := X_t^\tau$ is chosen to be independent of the delivery time~$\tau$, so is the forward measure~$Q := Q^\tau$.
\end{remark}

%===== FUTURES =====%
\subsection{Futures contracts} \label{section:futures}
As discussed in Section~\ref{section:Introduction} the forward kernel can be used to compute the price of futures contracts. In the following we assume the interest rate to equal $r = 0$ for notational convenience. Of course, when one assumes $r \neq 0$, discounting has to be taken into account. In Remark~\ref{remark:Discounting} we have some notes on how to change our framework to include discounting. Furthermore, we assume that all prices are normalized, meaning that we assume all prices to be in Euro/MWh as usual.

\begin{definition}[Futures contract price] \label{definition:futuresprice}
For $0 \leq t \leq \tau_1 < \tau_2$ we call 
\[
F_t (\tau_1, \tau_2) :=   \frac{1}{\tau_2 - \tau_1} \int_{\tau_1}^{\tau_2} f_t(u) \, du
\]
the price of a futures contract at time~$t$ delivering 1~MW continuously from~$\tau_1$ to~$\tau_2$.
\end{definition}

Since we denote all prices in Euro/MWh, the price that one pays at time~$t$ when one buys a futures contract delivering 1 MW from~$\tau_1$ to $\tau_2$ is given by $(\tau_2 - \tau_1) \, F_t(\tau_1, \tau_2)$, where we assume that $\tau_2 - \tau_1$ is measured in hours.

\begin{example}[Day-ahead spot price] \label{example:dayaheadspotprice}
We compute the day-ahead spot price as a futures contract. It is auctioned at day~$d-1$ at hour~$a$ and delivered at day~$d$ from $h$:$00$ until $(h+1)$:$00$ o'clock, i.e. 
\[
S(d, h) := F_{t_{d-1}^a}\left(t_d^h, t_d^{h+1}\right).
\]
Here $t_d^h$ denotes the time at day~$d$ and hour~$h$. 
\end{example}

The next theorem shows that the framework is consistent with cascading.\footnote{By cascading we mean the way how futures with a longer delivery period are settled. For example, a calendar year futures contract cascades (or splits up) into three monthly futures~(January, February, and March) and three quarterly futures~(Q2, Q3, and Q4) upon start of delivery. This way, these can be traded independently again. In the German market monthly futures do not cascade. However, the settlement price at the end of the delivery is exactly the average of the day-ahead spot prices during delivery. This could be interpreted that also monthly futures are cascading to the hourly (day-ahead) spot contracts, since their price converges to this average.} It also shows that there are no arbitrage opportunities in the sense that the cost of a futures contract delivering for one year is the same as the cost of its four quarters, for example.

\begin{proposition}[Consistency of cascading] \label{theorem:noarbitrage}
Let $0 \leq \tau_0 < \tau_1 < \tau_2 < \dots < \tau_n$ be delivery times, then we have
\[
(\tau_n - \tau_0)  \, F_t(\tau_0,\tau_n) = \sum_{i=1}^n (\tau_i - \tau_{i-1})  \, F_t(\tau_{i-1}, \tau_i)
\]
for all $t \geq 0$.
\end{proposition}
\begin{proof}
This follows directly from Definition~\ref{definition:futuresprice} and the countable additivity of the Lebesgue integral.
\end{proof}

\begin{lemma}
Fix $0 \leq t < \tau$. If $u \mapsto f_t(u)$ is almost surely continuous on $(\tau-\epsilon, \tau]$ for some $\epsilon > 0$, then we have
\[
\lim_{s \to \tau^-} F_t(s, \tau) = f_t(\tau)
\]
almost surely.
\end{lemma}
\begin{proof}
We compute
\[
\lim_{s \to  \tau^-} F_t (s, \tau) =  \frac{ \lim_{s \to  \tau^-}  \int_{s}^{\tau}\, f_t(u) \, du}{\lim_{s \to  \tau^-} \tau - s} 
=  \frac{ \lim_{s \to  \tau^-} - f_t(s)}{-1} = f_t(\tau)
\]
where we used L'H\^{o}pital's rule for the second equality.
\end{proof}
The previous lemma shows that the price of a futures contract delivering for just an instant equals the forward kernel. This supports the naming of the quantity~$f_t(\tau)$ as forward kernel. 

\begin{lemma}
Assume that the price forward curve~$\tau \mapsto f_0(\tau)$ is continuous. The futures price process~$F(\tau_1, \tau_2) := \{ F_t(\tau_1, \tau_2); t \geq 0 \}$ is a martingale. Its expectation is given by
\[
\E F_t(\tau_1, \tau_2) = F_0(\tau_1, \tau_2) = \frac{1}{\tau_2 - \tau_1} \int_{\tau_1}^{\tau_2} f_0(u) \, du
\]
for all $0 \leq t \leq \tau_1 < \tau_2$.
\end{lemma}
\begin{proof}
Since the price forward curve is continuous, it is bounded on any compact set, in particular intervals of the form~$[\tau_1, \tau_2]$, and therefore integrable on compacts. Direct computation with Fubini's Theorem shows that for~$0 \leq t < s$
\[
\E [ F_s(\tau_1, \tau_2) \, | \, \F_t] = \E  \left[\frac{1}{\tau_2 - \tau_1} \int_{\tau_1}^{\tau_2} f_s(u) \, du \, | \, \F_t \right]= \frac{1}{\tau_2 - \tau_1} \int_{\tau_1}^{\tau_2} \E[ f_s(u) \, | \, \F_t] \, du,
\]
where the latter exists and therefore all integrals exist. Combination with Lemma~\ref{lemma:expectationinstantaneousforward} now proves the theorem.
\end{proof}

\begin{remark}[$r\neq0$] \label{remark:Discounting}
If we assume that $r \neq 0 $, the futures price depends on the settlement date. There are two possibilities: settlement takes place either through continuous payments\footnote{Continuous settlement of the futures contract makes it more like a swap contract on the forward kernel.} during the delivery period or at once at the end of the delivery period. If $d_t(\tau)$ denotes the discount factor of a future payment at time~$\tau$ to an earlier time~$t$, the price of a futures contract is given by
\[
F_t (\tau_1, \tau_2) =   \frac{1}{\int_{\tau_1}^{\tau_2} d_t(u) \, du} \int_{\tau_1}^{\tau_2} d_t(u) \, f_t(u) \,du
\]
for continuous settlement and by
\[
F_t (\tau_1, \tau_2) =   \frac{1}{ (\tau_2 - \tau_1) \, d_t(\tau_2)} \int_{\tau_1}^{\tau_2} d_t(u) \, f_t(u) \,du
\]
for settlement at the end of delivery.
\end{remark}

%===== OPTIONS =====%
\subsection{Options on futures contracts} \label{section:options}
In this section we assume that the market noise is given by a geometric Brownian motion~(GBM) without drift, i.e.
\[
dX_t^\tau =  X_t^\tau \, \Sigma(t,\tau)^T  \, dW_t
\]
where $\Sigma(t,\tau)$ is a deterministic $m$-dimensional volatility vector and $W_t$ is an $m$-dimensional Brownian motion. The strong solution of $X_t^\tau$ is given by
\[
X_t^\tau = \exp\left(\int_0^t \Sigma(u, \tau)^T \, dW_u - \frac{1}{2} \int_0^t \Sigma(u, \tau)^T  \Sigma(u,\tau) \, du\right).
\]
In this case, $X_t^\tau$ satisfies Assumption~\ref{assumption:marketnoise} if $\Sigma(u,\tau)$ is square integrable in $u$. But this is already a requirement for the stochastic integral to be defined.

\begin{example}[Hull-White market noise dynamics]
A possible choice for $\Sigma$ is a two-factor forward dynamic similar to~\cite{Kiesel2009}, which is also discussed in a geometric setting by~\cite{Fanelli2018} for pricing options on futures. This volatility structure is extended by \cite{Latini2018} in an additive setting. They discussed a 2-factor volatility structure comparable to the two-factor Hull-White model for interest rate modelling~\cite[Section 4.2.5]{Brigo2006}. It is given by
\[
\Sigma(t, \tau)^T := (e^{-\kappa(\tau - t)} \sigma_1, \sigma_2(\tau) ),
\]
where $\sigma_1 > 0$ is the additional short-term volatility, $\kappa > 0$ is the rate of decay of the short-term volatility, and $\sigma_2(\tau)> 0$ is the long-term volatility at delivery time~$\tau$. A convenient choice for $\sigma_2$ is a piecewise constant function, being constant on delivery periods of tradable futures contracts. An advantage of this choice is that we can use the calibration methods for $X_t^\tau$ as discussed by \cite{Kiesel2009,Latini2018,Fanelli2018}. 
\end{example}

Throughout the rest of this subsection we assume that the conditional expectation of the structural component decomposes into an affine structure:
\begin{definition}[Affine structural component decomposition] \label{definition:affinestructuralcomponentdecomposition}
We say the structural component~$(g, Y_t)$ allows for the \emph{affine structural component decomposition}, if there exist deterministic functions $(t,\tau) \mapsto A_t^\tau \in \R^{n \times n}$ and  $(t,\tau) \mapsto B_t^\tau \in \R^n$ such that the following decomposition holds
\begin{equation} \label{eq:affinestructuralcomponentdecomposition}
\E[ g(Y_\tau) \, | \, \F_t] = g(A_t^\tau \, Y_t + B_t^\tau) 
\end{equation}
a.s. for all $\tau \geq t \geq 0$.
\end{definition}

This decomposition can be motivated by the fact that our best guess at time~$t$ for the state of market~$Y_\tau$ at time $\tau$ is an affine transformation of the current state of the market~$Y_t$. This is also the main idea behind Kalman filtering, for example. If the decomposition holds, this merely states that this best guess should hold under the transformation~$g$, which transforms the market state into a price.

It follows immediately that the forward kernel is given by
\begin{equation} \label{eq:forwardkernelascd}
f_t(\tau) =  X_t^\tau \,  g(A_t^\tau \, Y_t + B_t^\tau),
\end{equation}
when the affine structural component decomposition assumption is satisfied. Furthermore, the futures price of Definition~\ref{definition:futuresprice} can be rewritten as
\[
F_t (\tau_1, \tau_2)= \frac{1}{\tau_2 - \tau_1} \int_{\tau_1}^{\tau_2} X_t^u \, g( A_t^u \, Y_t + B_t^u ) \, du
\]
for all $0 \leq t \leq \tau_1 < \tau_2$. As immediate consequences we obtain:

\begin{lemma}
If $(g,Y_t)$ allows for the affine structural component decomposition, then $\E[ g(Y_\tau) \, | \, \F_t] = \E[ g(Y_\tau) \, | \, Y_t]$.
\end{lemma}

%\begin{lemma}[Moments of $f_t(\tau)$] \label{lemma:momentforwardkernel}
%If $(g,Y_t)$ allows for the affine structural component decomposition, then with $k \in \mathbb{N}$ the $k$th moment of~$f_t(\tau)$ is given by
%\[
%\E f_t(\tau)^k = \E(X_t^\tau)^k \, \E g(A_t^\tau \, Y_t + B_t^\tau)^k
%\]
%for all $\tau \geq t \geq 0$.
%\end{lemma}
%
%\begin{corollary}[Variance]
%Assume $(g,Y_t)$ allows for the affine structural component decomposition. The variance of $f_t(\tau)$ is given by
%\[
%\Var f_t(\tau) =  \E(X_t^\tau)^2 \, \E g(A_t^\tau \, Y_t + B_t^\tau)^2 -  f_0(\tau)^2
%\]
%for all $\tau \geq t \geq 0$.
%\end{corollary}

\begin{lemma}
Under assumption of the decomposition of Definition~\ref{definition:affinestructuralcomponentdecomposition} the forward kernel conditioned on $Y_t$ is lognormally distributed, i.e.
\[
(f_t(\tau) \, | \, Y_t = y) \sim LN\left(\ln[ g(A_t^\tau \, y + B_t^\tau)], \int_0^t \Sigma(u, \tau)^T  \Sigma(u,\tau) \, du\right).
\]
\end{lemma}
\begin{proof}
Using Equation~\eqref{eq:forwardkernelascd} we compute
\[
P\left(f_t(\tau) \leq x \, | \, Y_t = y\right) = P\left( X_t^\tau \, g(A_t^\tau \, y + B_t^\tau) \leq x \right),
\]
which shows the result since $X_t^\tau \sim LN\left(0,  \int_0^t \Sigma(u, \tau)^T  \Sigma(u,\tau) \, du \right)$.
\end{proof}

\begin{theorem} \label{theorem:optionsmomentsfutures}
If $(g,Y_t)$ allows for the affine structural component decomposition, then the first two moments of the futures price $F_t(\tau_1, \tau_2)$ exist and are given by
\[
\E [ F_t (\tau_1, \tau_2) \, | \, Y_t = y] = \frac{1}{\tau_2 - \tau_1} \int_{\tau_1}^{\tau_2} g( A_t^u \,y + B_t^u ) \, du
\]
and
\[
\E [F_t (\tau_1, \tau_2)^2 \, | \, Y_t =y ] = \frac{1}{(\tau_2 - \tau_1)^2} \int_{\tau_1}^{\tau_2} \int_{\tau_1}^{\tau_2} w^X_t(u,s)  \, w^Y_t(u,s, y) \, du \, ds,
\]
where
\begin{equation} \label{eq:optionscoexpectationmarketnoiseX}
w^X_t(u,s) := \exp \left( \frac{1}{2}  \int_0^t  \left(\Sigma(v, u)^T  \Sigma(v,s) + \Sigma(v, s)^T  \Sigma(v,u) \right) \, dv \right)
\end{equation}
and
\begin{equation} \label{eq:coefficientstructuralcomponentY}
w^Y_t(u,s, y) :=  g( A_t^u \,y + B_t^u ) \, g( A_t^s \,y + B_t^s ).
\end{equation}
\end{theorem}
\begin{proof}
We see that the expectation follows immediately by an Fubini argument combined with the fact that~$\E X_t^\tau = 1$ for all $\tau\geq 0$. Applying Fubini twice we find
\[
\E [F_t (\tau_1, \tau_2)^2 \, | \, Y_t = y ]  = \frac{\int_{\tau_1}^{\tau_2} \int_{\tau_1}^{\tau_2}  \E[X_t^u \, X_t^s] \, \E[ Y_u \, Y_s \, | \, Y_t =y ] \, du \, ds}{(\tau_2 - \tau_1)^2} ,
\]
where it is easy to verify that the expectations equal $\E[X_t^u \, X_t^s] = w^X_t(u,s)$ and $ \E[ Y_u \, Y_s \, | \, Y_t =y ] = w^Y_t(u,s,y)$ using Equation~\eqref{eq:affinestructuralcomponentdecomposition}.
\end{proof}

\begin{corollary} \label{corollary:optionsvariancefutures}
If $(g,Y_t)$ allows for the affine structural component decomposition, then the conditional variance of the futures price $F_t(\tau_1, \tau_2)$ is given by
\[
\Var [F_t (\tau_1, \tau_2) \, | \, Y_t = y] = \frac{1}{(\tau_2 - \tau_1)^2} \int_{\tau_1}^{\tau_2} \int_{\tau_1}^{\tau_2} \left(w^X_t(u,s) -1\right)  \, w^Y_t(u,s, y) \, du \, ds,
\]
where $w^X$ and $w^Y$ are given by Equation~\eqref{eq:optionscoexpectationmarketnoiseX} and Equation~\eqref{eq:coefficientstructuralcomponentY}, respectively.
\end{corollary}
\begin{proof}
We directly compute
\[
\Var [F_t (\tau_1, \tau_2) \, | \, Y_t = y]  = \E [F_t (\tau_1, \tau_2)^2 \, | \, Y_t =y ] - \E [ F_t (\tau_1, \tau_2) \, | \, Y_t = y]^2.
\]
Using Theorem~\ref{theorem:optionsmomentsfutures} the first term is immediately given and the second term can be computed using Fubini's Theorem
\begin{align*}
\E [ F_t (\tau_1, \tau_2) \, | \, Y_t = y]^2 &= \left( \frac{1}{\tau_2 - \tau_1} \int_{\tau_1}^{\tau_2} g( A_t^u \,y + B_t^u ) \, du\right)^2 \\
&= \frac{1}{(\tau_2 - \tau_1)^2} \int_{\tau_1}^{\tau_2} \int_{\tau_1}^{\tau_2} g( A_t^u \,y + B_t^u ) \, g( A_t^s \,y + B_t^s ) \, du \, ds,
\end{align*}
from which the result follows.
\end{proof}

\begin{remark}[Lognormal approximation] \label{remark:lognormalapproximationforfutures}
Similar to the discrete approach used by \cite{Kiesel2009} we have that the futures price is an integral of lognormally distributed variables, which can be approximated by a lognormal random variable with the same mean and standard deviation. Since there is no simple expression for the convolution of lognormal distributions, this approximation of the integral (or sum) of lognormal random variables is widely used in finance, e.g. in the context of LIBOR market models by~\cite{Brigo2006}. An analysis of this approximation, also with regard to Asian options (which may be compared to an option on a futures with delivery period), is found in~\cite{Dufresne2004}, for example.
\end{remark}

\begin{assumption}[Lognormal approximation] \label{assumption:lognormalapproximation}
Assume that the first two moments of the futures price $F_t(\tau_1, \tau_2)$ exist. Justified by Remark~\ref{remark:lognormalapproximationforfutures}, we then assume that
\[
(F_t(\tau_1, \tau_2) \, | \, Y_t = y) \approx ( \tilde{F}_t(\tau_1, \tau_2) \, | \, Y_t = y) \sim LN\left(\mu_F(y), \sigma_F^2(y) \right),
\]
i.e. the futures price is approximately lognormally distributed.
\end{assumption}

As stated in Remark~\ref{remark:lognormalapproximationforfutures} we need that the first two moments of $F$ and $\tilde{F}$ match, which is resolved by the following lemma:
\begin{lemma} \label{lemma:lognormalparametersforoptionsprice}
If $(g,Y_t)$ allows for the affine structural component decomposition and Assumption~\ref{assumption:lognormalapproximation} holds, then the mean and standard deviation of the lognormal distribution are given by 
\[
\mu_F(y) := \ln  \int_{\tau_1}^{\tau_2} g( A_t^u \,y + B_t^u ) \, du - \ln (\tau_2 - \tau_1)- \frac{1}{2} \sigma_F^2(y)
\]
and
\[
\sigma_F^2(y) := \ln  \left( 1+ \frac{\int_{\tau_1}^{\tau_2} \int_{\tau_1}^{\tau_2} \left(w^X_t(u,s) -1\right)  \, w^Y_t(u,s, y) \, du \, ds}{\int_{\tau_1}^{\tau_2} \int_{\tau_1}^{\tau_2} w^Y_t(u,s, y) \, du \, ds} \right),
\]
where $w^X$ and $w^Y$ are given by Equation~\eqref{eq:optionscoexpectationmarketnoiseX} and Equation~\eqref{eq:coefficientstructuralcomponentY}, respectively.
\end{lemma}
\begin{proof}
For a lognormal random variable $Z \sim LN(m, s)$, the expectation and variance are given by $\E Z = \exp(m + s^2/2)$ and $\Var Z  = (\E Z)^2 (\exp(s^2) - 1)$. Using Theorem~\ref{theorem:optionsmomentsfutures} and Corollary~\ref{corollary:optionsvariancefutures} the result is found by inverting these equations.
\end{proof}

Using this lemma we can compute the price (conditioned on $Y_t$) of call (and put) options on futures contracts by the Black-Scholes formula. A call option with strike price~$K$ and maturity~$T < \tau_1$ has a pay-off equal to 
\begin{equation} \label{eq:payoffoption}
(\tau_2 - \tau_1) \, (F_T(\tau_1, \tau_2) - K)^+.
\end{equation}
Recall that, as stated in Section~\ref{section:futures}, the price one has to pay for a futures contract at time~$T$ equals $(\tau_2 - \tau_1)  \,F_T(\tau_1, \tau_2)$, since we consider normalized prices.

\begin{proposition}[Conditional call option price] \label{proposition:calloptionsprice}
Assume that $(g,Y_t)$ allows for the affine structural component decomposition and let Assumption~\ref{assumption:lognormalapproximation} hold. Denote the futures price at maturity by $F := F_T(\tau_1, \tau_2)$. Let $\mu_F$ and $\sigma_F$ be given by Lemma~\ref{lemma:lognormalparametersforoptionsprice}. The price of a call option at $t=0$ with pay-off given by~\eqref{eq:payoffoption} conditioned on~$Y_T = y$ equals
\[
C_0(T,K, \tau_1, \tau_2; y) =  \Phi(\delta_1(y)) \int_{\tau_1}^{\tau_2} g( A_T^u \,y + B_T^u ) \, du \, - (\tau_2 - \tau_1)  \, K \, \Phi(\delta_2(y)),
\]
where $\Phi$ is the cumulative distribution function of the standard normal distribution,
\[
\delta_2(y) := \frac{\mu_F(y) - \ln K}{\sigma_F(y)},
\]
and $\delta_1(y) := \delta_2(y) + \sigma_F(y)$.
\end{proposition}
\begin{proof}
Using the discounted conditional expectation of the pay-off given in~\eqref{eq:payoffoption} yields
\begin{align*}
\frac{C_0(T,K, \tau_1, \tau_2; y)}{\tau_2 - \tau_1 } &= \E\left[ (F - K)^+ \, | \, Y_T = y \right]\\
&=  \E[ F \, \ind_{\{F \geq K \}} \, | \, Y_T = y ] - K \, P(F \geq K \, | \, Y_T = y) ,
\end{align*}
where noting that we have $(F \, | \, Y_T = y) \sim LN(\mu_F(y), \sigma_F^2(y))$, yields the result by direct computation.
\end{proof}

As an immediate consequence we have:
\begin{corollary}[Call option price] \label{corollary:calloptionsprice2}
Assume that $(g,Y_t)$ allows for the affine structural component decomposition and let Assumption~\ref{assumption:lognormalapproximation} hold. Let $\mu_F$ and $\sigma_F$ be given by Lemma~\ref{lemma:lognormalparametersforoptionsprice}. The price of a call option at $t=0$ with pay-off given by~\eqref{eq:payoffoption} equals
\begin{equation}\label{eq:realcalloptionprice}
C_0(T,K, \tau_1, \tau_2) = \E C_0(T,K, \tau_1, \tau_2; Y_T),
\end{equation} 
where the conditional call option price~$C_0(T,K, \tau_1, \tau_2; y)$ is given in Proposition~\ref{proposition:calloptionsprice}.
\end{corollary}

When the distribution of~$Y_T$ is specified the price of a call option given by Equation~\eqref{eq:realcalloptionprice} might be evaluated analytically, numerically, or through simulative methods such as Monte Carlo estimation. Alternatively, with further assumptions on the distribution of~$Y_T$ this expectation could also be approximated differently.

%%===== OVERVIEW =====%
\subsection{Model representation of exchange traded products} \label{section:OverviewContracts}
In this section we give an overview of the prices of several different electricity contracts in this HJM framework. Although there is not a single unique quoted continuous electricity price we regard $F_t(\tau_1, \tau_2)$ as the true fair price for the delivery period from $\tau_1$ to $\tau_2$ at any trading time~$t$.

\paragraph{Futures price}
The price of a futures contract at time~$t$ delivering 1~MW continuously from~$\tau_1$ to~$\tau_2$ is given by Definition~\ref{definition:futuresprice} and denoted by $F_t(\tau_1,\tau_2)$.

\paragraph{Options on futures}
In the setting of Section~\ref{section:options} the price of call and put options on futures contracts can be computed by the Black-Scholes formula as given by Proposition~\ref{proposition:calloptionsprice} or Corollary~\ref{corollary:calloptionsprice2}.

\paragraph{Day-ahead spot prices}
The day-ahead spot price equals the futures price within this framework as discussed in Example~\ref{example:dayaheadspotprice}.

\paragraph{$\text{ID}_1$ and $\text{ID}_3$ price}
The $\text{ID}_1$ and $\text{ID}_3$ price indices on the German intraday market are given as the one and three hour volume-weighted average of all intraday trades before delivery. Therefore, we suggest the $\text{ID}_n$ price for the delivery period from $\tau_1$ to $\tau_2$ to equal
\[
\text{ID}_n(\tau_1, \tau_2) :=  \frac{2}{2n -1} \int_{\tau_1 - n}^{\tau_1 - 0.5} F_{u}(\tau_1, \tau_2) \, du,
\]
where $n=1 $ or $n=3$ and the subtraction of $\tau_1$ is meant in hours.

%===== APPLICATIONS =====%
\section{Examples of the structural component} \label{section:applications}
First we show how two classical day-ahead spot price models can be used in this HJM framework. Then we also introduce a structural model approach as well as a multi-factor model approach for $Y$.

To make defining a model easier in this framework we introduce the relative structural component, which can be used to set the initial price forward  curve~(PFC) to an existing one:
\begin{definition}[Relative structural component] \label{definition:RelativeStructuralComponent}
The additive mean-normalized version of $g(Y_\tau)$
\[
I^a_\tau := g(Y_\tau) - \E g(Y_\tau)
\]
is called the \emph{additive relative structural component} and its multiplicative mean-normalized version
\[
I^m_\tau := \frac{g(Y_\tau)}{\E g(Y_\tau)}
\]
is called the \emph{multiplicative relative structural component}.
\end{definition}

We directly obtain from these definitions:
\begin{corollary}
The relative structural components~$I^a$ and $I^m$ are stochastic processes with constant expectation $\E I^a_\tau = 0$ and $\E I^m_\tau = 1$ for all $\tau \geq 0$.
\end{corollary}

\begin{corollary}[Arithmetic PFC decomposition] \label{corollary:additivepfcdecomposition}
For a given initial price forward curve~$f_0(\tau)$ the forward kernel equals
\[
f_t(\tau) =  X_t^\tau \, \left( f_0(\tau) + \E[ I^a_\tau \, | \, \F_t ] \right),
\]
where $I^a_\tau$ is the arithmetic relative structural component given in Definition~\ref{definition:RelativeStructuralComponent}.
\end{corollary}
\begin{proof}
Define an extended structural component $\tilde{Y}_\tau = (Y_\tau, f_0(\tau)) \in R^{n+1}$, where $f_0(\tau)$ is the constructed PFC, and another function~$\tilde{g}(y, x) = x + g(y) - \E g(y)$. It is clear that $\tilde{Y}$ and $\tilde{g}$ satisfy Assumption~\ref{assumption:marketimbalance}. It follows immediately that $\tilde{g}(\tilde{Y}(\tau)) = f_0(\tau) + I^a_\tau$, which proves the result.
\end{proof}

\begin{corollary}[Geometric PFC decomposition] \label{corollary:multiplicativepfcdecomposition}
For a given initial price forward curve~$f_0(\tau)$ the forward kernel equals
\[
f_t(\tau) = f_0(\tau) \, X_t^\tau \, \E[ I^m_\tau \, | \, \F_t ],
\]
where $I^m_\tau$ is the geometric relative structural component given in Definition~\ref{definition:RelativeStructuralComponent}.
\end{corollary}
\begin{proof}
The result can be shown analogously to the proof of Corollary~\ref{corollary:additivepfcdecomposition}.
\end{proof}

The interpretation of these decompositions is that today's price forward curve is the expectation of the forward kernel that is being disturbed by the market noise~$X_t^\tau$ in trading time~$t$ and by the structural component in delivery time~$\tau$. Depending on the choice of the structural component~$(g, Y_\tau)$ this disturbance can be chosen to be multiplicatively in case of the geometric PFC decomposition or additively in case of the arithmetic PFC decomposition.

\subsection{Classical spot models}
We can use classical day-ahead spot price models in our framework by choosing $g(Y_t) = S_t$, where $S_t$ denotes the spot price at time~$t$. Two examples of spot price models that we explicitly compute in this section are the spot price models by \cite{Schwartz2000} and \cite{Lucia2002}.

For both examples we need the same structural component and therefore we assume in this subsection that it is given by $Y_t = (y_t^1, y_t^2) \in \R^2$. The first process is an Ornstein-Uhlenbeck process, i.e.
\begin{equation} \label{eq:schwartzfactorone}
dy^1_t = -\kappa \, y^1_t \, dt + \sigma_1 \, dW_t^1, \quad y^1_0 = 0,
\end{equation}
and the second
\begin{equation}  \label{eq:schwartzfactortwo}
y^2_t = \mu_2 t + \sigma_2 \rho W^1_t + \sigma_2 \sqrt{1-\rho^2} W^2_t
\end{equation}
is a (correlated) Brownian motion with drift. The standard one-dimensional Brownian motions~$W^1$ and $W^2$ are assumed to be independent. The parameters $\kappa >0$, $\sigma_1,\sigma_2 > 0$, $-1 \leq \rho \leq 1$, and $\mu \in \R$ are assumed to be real-valued.

\begin{example}[Schwartz and Smith] \label{example:schwartzsmithmodel}
\cite{Schwartz2000} define the day-ahead spot price using the function $g(y_1, y_2) = e^{y_1 + y_2}$, i.e. they chose the price to equal $S_t := g(Y_t) = \exp(y^1_\tau + y^2_\tau)$. In the HJM framework this transfers to the following forward kernel
\[
f_t(\tau) = X_t^\tau \, \E[e^{y^1_\tau + y^2_\tau} \, | \, \F_t],
\]
where we do not assume any extra conditions on $X^\tau$ apart from Assumption~\ref{assumption:marketnoise}.

In this setting we can explicitly compute the conditional expectation on $g(Y_\tau)$ and we find
\begin{align*}
\ln \E[e^{y^1_\tau + y^2_\tau} \, | \, \F_t] &= e^{-\kappa (\tau -t)} y^1_t + y^2_t +   \left( \mu_2+  \tfrac{\sigma_2^2}{2} \right) (\tau - t)  \\
& \quad + \frac{\sigma_1^2}{4 \kappa} \left(1-e^{-2\kappa(\tau-t)} \right) + \frac{\rho \, \sigma_1 \sigma_2 }{\kappa} \left(1-e^{-\kappa(\tau-t)} \right).
\end{align*}
This implies that this model for $g$ and $Y_\tau$ satisfies the affine structural component decomposition of Definition~\ref{definition:affinestructuralcomponentdecomposition}. The coefficient~$A_t^\tau$ of the decomposition is given by
\begin{equation} \label{eq:ASCDforSchwartzSmith}
A_t^\tau = 
\begin{pmatrix}
e^{-\kappa (\tau-t)} & 0 \\
0 & 1
\end{pmatrix}
\end{equation}
and $B_t^\tau$ can be chosen to be any vector in $\R^2$ such that 
\[
\ln g(B_t^\tau) = \left( \mu_2+  \frac{\sigma_2^2}{2} \right) (\tau - t)   + \frac{\sigma_1^2}{4 \kappa} \left(1-e^{-2\kappa(\tau-t)} \right) + \frac{\rho \, \sigma_1 \sigma_2 }{\kappa} \left(1-e^{-\kappa(\tau-t)} \right) 
\]
holds.

Since the function $g$ is multiplicative in nature, the geometric PFC decomposition, Corollary~\ref{corollary:multiplicativepfcdecomposition}, is especially suited for this model. The conditional expectation of the multiplicative relative structural component is given by
\begin{align*}
\ln \E[ I^m_\tau \, | \, \F_t ] &= \ln \frac{g(A_t^\tau \, Y_t + B_t^\tau)}{\E g(Y_\tau)} \\
&=  e^{-\kappa (\tau -t)} y^1_t + y^2_t - \left( \mu_2+  \tfrac{\sigma_2^2}{2} \right) t \\
&\quad + \frac{\sigma_1^2 e^{-2\kappa\tau}}{4 \kappa} \left(1-e^{2\kappa t} \right) + \frac{\rho \, \sigma_1 \sigma_2 e^{-\kappa \tau}}{\kappa} \left(1-e^{\kappa t} \right),
\end{align*}
and the forward kernel decomposes to
\[
f_t(\tau)= f_0(\tau) \, X_t^\tau \, e^{e^{-\kappa (\tau -t)} y^1_t + y^2_t - \left( \mu_2+  \tfrac{\sigma_2^2}{2} \right) t  + \frac{\sigma_1^2 e^{-2\kappa\tau}}{4 \kappa} \left(1-e^{2\kappa t} \right) + \frac{\rho \, \sigma_1 \sigma_2 e^{-\kappa \tau}}{\kappa} \left(1-e^{\kappa t} \right)},
\]
where any initial price forward curve~$f_0(\tau)$ can be used.
\end{example}

\begin{example}[Lucia and Schwartz]
\cite{Lucia2002} discuss four different models. Here, we highlight the arithmetic two factor model for the spot price. This model is defined by the function~$g(y_1, y_2) = y_1 + y_2$ and the forward kernel equals 
\[
f_t(\tau) = X_t^\tau \,  \E[y^1_\tau + y^2_\tau \, | \, \F_t].
\]
Again, apart from Assumption~\ref{assumption:marketnoise} the process~$X^\tau$ can be chosen freely.

The conditional expectation can easily be computed as
\[
\E[y^1_\tau + y^2_\tau \, | \, \F_t] =  e^{-\kappa(\tau_t)} \, y^1_t + y^2_t + \mu_2 (\tau -t)
\]
and the affine structural component decomposition of Definition~\ref{definition:affinestructuralcomponentdecomposition} follows immediately with the coefficient~$A_t^\tau$ given by Equation~\eqref{eq:ASCDforSchwartzSmith} and $B_t^\tau$ can be any vector in $\R^2$ such that $g(B_t^\tau) = \mu_2 (\tau -t)$.

The additive nature of $g$ makes the arithmetic PFC decomposition, Corollary~\ref{corollary:additivepfcdecomposition}, the best suited candidate for this model. It follows that
\[
f_t(\tau) = X_t^\tau \, \left( f_0(\tau)  + e^{-\kappa(\tau_t)} \, y^1_t + y^2_t - \mu_2 t \right)
\]
for any initial price forward curve~$f_0(\tau)$. We continue the study of this type of forward kernel in Section~\ref{sec:factormodel} with a factor model approach.
\end{example}

In the rest of this section we will give two further examples of the structural component~$Y$. The first is based on the structural model approach for day-ahead spot prices and the other uses multi-factor models, which are the sum of Ornstein-Uhlenbeck type processes, cf. \cite{Benth2008}.

%===== STRUCTURAL =====%
\subsection{Structural model approach} \label{sec:structuralmodel} 
We will use the HJM framework to model the structural component by a structural model approach: a spot price modelling technique started by~\cite{Barlow2002} which uses the idea of equilibrium of supply and demand to derive a spot price. In contrast to reduced-form models which need to implement a jump component to model spikes, structural models use a non-linear transformation of a (Gaussian) diffusion process to reach this goal. This method has been developed further by many authors, e.g. \cite{Aid2009,Wagner2014}.

For the real-valued demand process~$D$ we use a Gaussian Ornstein-Uhlenbeck process, i.e.
\[
dD_t = -\lambda \, D_t \, dt + \sigma \, dW_t,  \quad D_0 = 0.
\]
We choose the structural component to equal
\[
Y_t :=
\begin{pmatrix}
\beta(t) \\
D_t
\end{pmatrix},
\]
where $\beta(t)$ is a real-valued deterministic function. Furthermore, we define the function~$g$ as follows
\[
g(y_1, y_2) = \gamma + y_1 \, \sinh(\alpha \, y_2) = \gamma + y_1  \frac{e^{\alpha \, y_2} - e^{- \alpha \, y_2}}{2}
\]
for $\alpha > 0$ and $\gamma > 0$. Through the first coordinate of $Y_t$, i.e. $\beta(t)$, we associate $y_1$ with the evolution of time and $y_2$ through the second coordinate of $Y_t$, namely $D_t$, with the demand. Therefore, $g(Y_t)$ represents the price at time~$t$ for a load of~$D_t$ through the \emph{merit order curve}.

\begin{remark}[Extension of the model]
It might be convenient to use more realistic models, such as described by \cite{Wagner2014}. This is an extension of the OU model, where stochastic processes for wind and solar infeed are subtracted from the demand process~$D$. This difference is seen to model power prices even more accurately. It can easily be seen that the structural component~$Y_t$ and function~$g$ can be extended for these processes.
\end{remark}

Using the auxiliary function~$\nu^2(s) := \frac{\sigma^2}{2 \lambda}(1 - e^{-2\lambda s})$ the affine structural component decomposition of Definition~\ref{definition:affinestructuralcomponentdecomposition} can be derived from the following theorem:
\begin{theorem}\label{theorem:structuralmodelconditionalexpectation}
The conditional expectation of the structural component is given by
\[
\E[g(Y_\tau) \, | \, \F_t ] = \gamma +  \beta(\tau)\, e^{\frac{\alpha^2 }{2}  \nu^2(\tau - t)} \sinh(\alpha \, e^{-\lambda (\tau-t)}  \, D_t)
\]
for all $\tau \geq t \geq 0$.
\end{theorem}
\begin{proof}
For Gaussian OU processes we have the following decomposition
\[
D_\tau \overset{\text{d}}{=} e^{-\lambda (\tau-t)} \, D_t +  \nu(\tau-t) \, \varepsilon, \quad \varepsilon \sim \mathcal{N}(0,1).
\]
Now, exploiting the decomposition and plugging it into the definition we get
\begin{align*}
\E[g(Y_\tau) \, | \, \F_t ] &=   \gamma +  \beta(\tau) \, \E\left[ \sinh(\alpha \, D_\tau) \, | \, \F_t \right]  \\
&= \gamma + \beta(\tau)  \sinh(\alpha \, e^{-\lambda (\tau-t)} \, D_t) \, \E\left[ e^{\alpha  \nu(\tau-t) \varepsilon}\right] \\
&= \gamma + \beta(\tau)\, e^{\frac{\alpha^2 }{2}  \nu^2(\tau - t)} \sinh(\alpha \, e^{-\lambda (\tau-t)}  \, D_t)
\end{align*}
by symmetry of the normal distribution.
\end{proof}

\begin{corollary}[Affine structural component decomposition]
With coefficients given by
\[
A_t^\tau = 
\begin{pmatrix}
\frac{\beta(\tau)}{\beta(t)} e^{\frac{\alpha^2 }{2} \nu^2(\tau - t)} & 0 \\
0 & \alpha e^{-\lambda (\tau-t)}
\end{pmatrix}
\]
and $B_t^\tau = 0 \in \R^2$ the affine structural component decomposition of Definition~\ref{definition:affinestructuralcomponentdecomposition} holds.
\end{corollary}

By Theorem~\ref{theorem:structuralmodelconditionalexpectation} it follows immediately by taking $t = 0$ that the expectation~$\E g(Y_\tau) = \gamma > 0$ for all $\tau \geq 0$. Therefore we can use both the additive and geometric PFC decomposition, i.e. Corollary~\ref{corollary:additivepfcdecomposition} and Corollary~\ref{corollary:multiplicativepfcdecomposition}, respectively. In the additive case the forward kernel equals
\begin{align*}
f_t(\tau) &= X_t^\tau \left( f_0(\tau) + g(A_t^\tau \, Y_t) - \gamma \right) \\
&=  X_t^\tau \left( f_0(\tau) + \beta(\tau)\, e^{\frac{\alpha^2 }{2}  \nu^2(\tau - t)} \sinh(\alpha \, e^{-\lambda (\tau-t)}  \, D_t) \right),
\end{align*}
whereas in the multiplicative case it equals
\begin{align*}
f_t(\tau) &= f_0(\tau) \, X_t^\tau \, \frac{g(A_t^\tau \, Y_t)}{\gamma} \\
&=  f_0(\tau) \, X_t^\tau \left( 1 + \frac{\beta(\tau) }{\gamma} \, e^{\frac{\alpha^2 }{2}  \nu^2(\tau - t)} \sinh(\alpha \, e^{-\lambda (\tau-t)}  \, D_t) \right).
\end{align*}
For both decompositions any initial price forward kernel can be used.

\subsection{Arithmetic factor model approach} \label{sec:factormodel}
In this section we use an arithmetic factor model approach for the structural component in the HJM framework. More precisely, the structural component is given by an $n$-dimensional L\'evy driven Ornstein-Uhlenbeck process
\[
dY_t = - \Lambda \, Y_t \, dt + dL_t, \quad Y_0  = y_0,
\]
where $\Lambda = \diag(\lambda_1, \lambda_2, \dots, \lambda_n) \in \R^{n\times n}$ with $\lambda_1, \lambda_2, \dots, \lambda_n >0$ and $L$ is an $n$-dimensional L\'evy process. For more information on this type of moving average process we refer the interested reader to~ \cite{Wolfe1982,Jurek1983,BarndorffNielsen2001,Applebaum2009,Sato2013}. For an application of OU processes in the form of multi-factor models for energy prices we refer to \cite{Benth2008}.

The function~$g$ is given by the summation of all the coefficients, i.e. we assume that $g(y) = \sum_{i=1}^n y_i$. If $Y_t$ satisfies Assumption~\ref{assumption:marketimbalance} we can explicitly compute the conditional expectation:
\begin{theorem} \label{theorem:factormodelconditionalexpectation}
The conditional expectation of the structural component is given by
\[
\E[ g(Y_\tau) \, | \, \F_t ] = g\left( e^{-\Lambda (\tau - t)} \, Y_t + \E \int_t^\tau e^{-\Lambda (\tau - u)} \, dL_u \right)
\]
for all $\tau \geq t \geq 0$.
\end{theorem}
\begin{proof}
For general OU processes the same decomposition holds as was used in the proof of Theorem~\ref{theorem:structuralmodelconditionalexpectation}, i.e.
\[
Y_\tau =  e^{-\Lambda (\tau - t)} \, Y_t + \int_t^\tau e^{-\Lambda (\tau - u)} \, dL_u.
\]
Noting that the first term is $\F_t$-measurable and the second term is independent of $\F_t$ yields the result, as the sum $g$ and $\E$ commute.
\end{proof}

As a direct consequence we obtain:
\begin{corollary}[Affine structural component decomposition]
With coefficients given by $A_t^\tau = e^{-\Lambda (\tau - t)}$ and $B_t^\tau = \E \int_t^\tau e^{-\Lambda (\tau - u)} dL_u$ the affine structural component decomposition of Definition~\ref{definition:affinestructuralcomponentdecomposition} holds.
\end{corollary}

Due to the additive structure of $g$ the logical PFC decomposition to choose in this setting is the arithmetic one, i.e. Corollary~\ref{corollary:additivepfcdecomposition}. From Theorem~\ref{theorem:factormodelconditionalexpectation} we find that the expectation is given by
\[
\E g(Y_\tau) = g \left(e^{-\Lambda \tau} \, y_0 + \E \int_0^\tau e^{-\Lambda ( \tau - u)} \, dL_u\right).
\]
It follows that the forward kernel is given by
\[
f_t(\tau) = X_t^\tau \left( f_0(\tau) + g\left( e^{-\Lambda (\tau - t)} \, Y_t  - e^{-\Lambda \tau} \, y_0 - \E \int_0^t e^{-\Lambda ( \tau - u)} \, dL_u \right) \right),
\]
where $f_0(\tau)$ can be any initial price forward curve.

\section{Conclusion} \label{section:Conclusion}
In this paper we have developed a unifying Heath-Jarrow-Morton~(HJM)  framework that
\begin{itemize}
\item models intraday, spot, and futures prices,
\item is based on two stochastic processes motivated by economic interpretations, %: the structural component driven by the market equilibrium and the market noise caused by incomplete market information,
\item separates the stochastic dynamics in trading and delivery time,
\item is consistent with the initial term structure (i.e. the price forward curve),
\item is able to price options on futures by means of the Black-Scholes formula,
\item allows for the use of classical day-ahead spot price models such as \cite{Schwartz2000,Lucia2002},
\item includes many model classes such as structural models and factor models.
\end{itemize}
%
%In this paper we introduced a flexible Heath-Jarrow-Morton~(HJM) framework that allows for consistent modelling of prices of intraday, spot, futures, and options on futures as shown in Section~\ref{section:HJMframework}. This framework is based on two stochastic processes with economic interpretations, namely: one factor that simulates the price that is induced by the equilibrium of supply and demand, and the second factor is interpreted as noise caused by incomplete market information, such as the uncertainty in weather predictions, and by the illiquidity of the electricity market. Another feature of this HJM framework is the consistence with the initial term structure given in the form of a price forward curve.
%
%Furthermore, the framework allows for day-ahead spot price models to be used in an HJM setting, a feature that immediately opens a wide range of applications. In Section~\ref{section:applications} we included several explicit examples of classical spot price models such as \cite{Schwartz2000,Lucia2002}. We also show how structural models (i.e. based on supply and demand) and factor models can be used within this framework. To the best of our knowledge, this is the first time structural HJM models are introduced in our context.
%
To further the development of this framework empirical studies are needed: statistical evaluations but also calibration methods need to be discussed. The theoretical applications of Section~\ref{section:applications} need to be specified and calibrated to real data from intraday, spot, futures, and option prices. This is subject of future research.

%===== ACKNOWLEDGMENT =====%
\section*{Acknowledgments}
WJH is grateful for the financial support from Fraunhofer ITWM (\emph{Fraunhofer Institute for Industrial Mathematics ITWM}, \url{www.itwm.fraunhofer.de}).

\bibliographystyle{apacite}
\bibliography{references}

%===== FOOTNOTES =====%
%\theendnotes
\end{document}